\documentclass[runningheads]{llncs}

\usepackage[T1]{fontenc}

\usepackage{times}
\usepackage{graphicx}
\usepackage{amsmath}
\usepackage{amssymb}
\usepackage{mathtools}
\usepackage{todonotes}[disable]
\usepackage{mathrsfs}
\usepackage[colorlinks=true, linkcolor=red, urlcolor=blue, citecolor=green]{hyperref}
\usepackage{subcaption}
\usepackage{pdfpages}
\usepackage{svg}
\usepackage{thmtools}
\usepackage{thm-restate}

\usepackage[inline]{enumitem}

\newcommand{\N}[1]{N_{#1}} 

\newcommand{\Sv}{\mathbf{S}} 
\newcommand{\SvInit}{\Sv^0} 
\newcommand{\Si}[1]{\Sv_{#1}} 
\newcommand{\St}[2]{\Sv^{#2}_{#1}} 

\newcommand{\Bv}{\mathbf{B}} 
\newcommand{\BvInit}{\Bv^0} 
\newcommand{\Bi}[1]{\Bv_#1} 
\newcommand{\Bt}[2]{\Bv^{#2}_{#1}} 
\newcommand{\PBv}[1]{\mathtt{pub}\Bv^{#1}}

\newcommand{\Tol}{\tau} 
\newcommand{\Toli}[1]{\Tol_{#1}} 

\newcommand{\Maj}{\mathcal{M}}

\newcommand{\SOMminus}{\mbox{SOM}^-}
\newcommand{\SOMplus}{\mbox{SOM}^+}

\newcommand{\Prime}[1]{\bar{#1}}

\newcommand{\Nat}{\mathbb{N}}

\renewenvironment{proof}%
  {\begin{trivlist} \item[\hskip \labelsep {\bfseries Proof.}]}%
  {\hfill \qed \end{trivlist}}

\begin{document}
    \title{The Sound of Silence in Social Networks 
    \thanks{This work has been partially supported by the SGR project PROMUEVA (BPIN 2021000100160) under the supervision of Colombian Ministry of Science, Technology and Innovation (Minciencias) and by the CNRS project TOBIAS under the MITI interdisciplinary program.
}} 

    \author{Jes\'us Aranda\inst{1} \and  Juan Francisco D\'iaz \inst{1} \and David Gaona\inst{1} 
        \and
        Frank Valencia\inst{2,3} }

    \institute{Universidad del Valle, Colombia 
     \and CNRS-LIX, \'Ecole Polytechnique de Paris, France
     \and Pontificia Universidad Javeriana Cali, Colombia
    }
    
    \maketitle 

    \begin{abstract}
We generalize the classic multi-agent DeGroot framework for opinion dynamics by incorporating the Spiral of Silence theory from political science, which posits that individuals may withhold their opinions when they perceive them to be in the minority. As in the original DeGroot model, the social network is represented as a weighted directed graph encoding how agents influence one another. However, agents holding minority opinions become \emph{silent}, meaning they do not express their views.

We introduce two families of models. In \emph{Silence Opinion Memoryless} ($\SOMminus$) models, agents update their opinions by averaging those of their \emph{non-silent} neighbors. In \emph{Silence Opinion Memory-based} ($\SOMplus$) models, agents average the opinions of \emph{all} neighbors, but for silent ones, only the most recently expressed opinion is used. We show that $\SOMminus$ models guarantee consensus on clique graphs but, unlike the classic DeGroot model, not on all strongly connected aperiodic graphs. For $\SOMplus$ models, even cliques may fail to reach consensus, illustrating that even minimal memory can significantly affect opinion dynamics. Finally, we validate our models through large-scale simulations on small-world networks with over \emph{two million} agents. The results support  the Spiral of Silence theory and reveal inherent limitations to consensus in more realistic settings.
\end{abstract}


    \section{Introduction}
     \label{sec:introduction}
     Social networks play a significant role in \emph{opinion formation} often having strong impact on polarization and the ability to reach \emph{consensus}. Broadly, the dynamics of opinion formation in social networks involve users expressing their opinions, being exposed to the opinions of others, and potentially adapting their own views based on these interactions. Modeling these dynamics enables us to glean insights into how opinions form and spread within social networks. 

The  DeGroot framework \cite{degroot}  is one of the most prominent formalisms for opinion formation and consensus-building in social networks. In the models of this framework, a social network is represented as a weighted directed graph, where edges denote the degree to which individuals (i.e., \emph{agents}) influence one another. Each agent holds an opinion, expressed as a value in $[0,1]$, indicating their level of agreement with an underlying proposition (e.g., ``\emph{AI is a threat to humanity}"). Agents repeatedly update their opinions by taking the weighted average of their opinion differences with those who influence them (i.e., their \emph{contacts}). There is empirical evidence validating the opinion formation through averaging of the model in controlled sociological experiments \cite{degrootEmpirico2015}. 

\emph{Consensus}, i.e., convergence to a common opinion, is a central property in models of social learning and opinion formation \cite{survey}. In fact, difficulties in achieving consensus are a sign of a polarized society. A fundamental result for (classic) DeGroot models shows that agents converge to consensus if the influence graph is strongly connected and aperiodic. The DeGroot framework continues to be a focus of research for constructing frameworks for understanding opinion formation dynamics in social networks (e.g,.~\cite{survey,Alvim2024,alvim21,alvim:hal-03740263,Mossel2017,alvim:hal-03740263,Generalize1,Generalize2,Generalize3,GeneralizeBF1,demarzo2003persuasion,chatterjee1977towards}).

Nevertheless, the classic DeGroot formalism makes an assumption that could be overly constraining within social network contexts. It assumes that \emph{all agents express their opinions at each time unit}. This assumption, which renders models tractable, may hold in some controlled scenarios as participants are often encouraged to express their views freely and consistently. However, in many real-world situations, some individuals may choose not to express their opinions due to personal choice or social pressure.

Indeed, the \emph{Spiral of Silence} \cite{Noelle-Neumann1974} is a well-established social theory that describes how individuals may be unwilling to express their opinions when they perceive themselves to be in the minority. This reluctance can lead to the \emph{reinforcement of dominant views} within a social network. The theory asserts that individuals have a natural tendency to avoid social isolation and seek acceptance within their social groups. When people believe their opinions are unpopular or likely to be met with disapproval, they may opt to remain silent.

The relevance of the Spiral of Silence has been  validated in  social media environments. Large-scale empirical studies show that \emph{social media platforms do not offer alternative spaces for minority opinion expression}, but instead reinforce silencing behaviors~\cite{Hampton2014}. Meta-analyses confirm that the link between perceived opinion support and willingness to speak out remains strong in digital contexts~\cite{Matthes2018}. Recent work finds that users are particularly sensitive to perceived disagreement among their online contacts, leading to self-censorship that can \emph{spill over into offline behavior}~\cite{Hampton2014,Sohn2022}. During the COVID-19 pandemic, individuals frequently conformed publicly to dominant narratives while privately holding dissenting views, further illustrating the theory’s enduring relevance~\cite{Atanesyan2023}. These findings suggest that the dynamics of opinion suppression described by The Spiral of Silence not only persist but may be exacerbated in digital social networks.

In this paper, we generalize DeGroot models into a framework where agents may choose to remain silent at a given time following the Spiral of Silence. We consider two possibilities, leading to the two families of models described below.

\emph{Memoryless framework $\SOMminus$.} In these models, silent agents are excluded from the opinion updates of the agents they would otherwise influence. Additionally, agents become silent at a given time  if their views do not align with the majority of their non-silent contacts. This framework is called the \emph{silence opinion memoryless } ($\SOMminus$) models, as the previous opinions of silent agents are not retained. This corresponds to a social scenario in which opinions (e.g., expressed in posts) are removed once they have been accessed.

Notice that ignoring silent agents at a given time unit amounts to removing certain edges from the influence graph at that time. Thus, a fundamental distinction from DeGroot models is that $\SOMminus$ models exhibit \emph{dynamic influence} as edges may disappear and reappear during opinion evolution.

\emph{Memory-based framework $\SOMplus$.} In these models, agents choose to be silent if their opinion does not align with the most recent public opinions of \emph{the majority of} their contacts. Unlike in $\SOMplus$ models, silent agents \emph{are not excluded} from the opinion updates of the agents they influence: when their current opinion is unknown, their most recent public opinion is taken into account in the update.  This framework is called the \emph{silence opinion memory-based} ($\SOMplus$) models, as the most recent opinion of each agent is retained.

 A property that distinguishes $\SOMplus$ models from classic  DeGroot models (and $\SOMminus$ models) is that the latter are \emph{Markovian} processes: The next state depends on the current state but not the past states. Thus, $\SOMplus$ models are \emph{history-dependent} but with very limited memory; only most recent public opinions are remembered. We will show that this minimal notion of memory has an impact on  consensus.

\emph{Contributions.}
We make the following  theoretical and experimental contributions:
\begin{enumerate}
    \item We generalize the DeGroot model to incorporate key aspects of the Spiral of Silence theory. To the best of our knowledge, this is the first extension of the DeGroot framework to do so.

    \item We show that in $\SOMminus$, convergence to consensus is guaranteed in clique graphs (i.e., fully connected graphs) with more than two agents. This result highlights that consensus remains possible under Spiral of Silence dynamics, even when silent agents are excluded from updates.

    \item We prove that, unlike in the classical DeGroot model, $\SOMminus$ does not guarantee consensus in all strongly connected aperiodic graphs.

    \item We demonstrate that $\SOMplus$ fails to guarantee consensus even in clique graphs. This negative result underscores that the limited memory in $\SOMplus$ can significantly alter the dynamics, making consensus harder to achieve under Spiral of Silence assumptions.

    \item We validate our models through examples and large-scale simulations involving over \emph{two million agents} on randomly generated networks with small-world properties typical of real social systems. These simulations support core claims of the Spiral of Silence theory, particularly the \emph{reinforcement of dominant views in social networks}. The simulation code is available at:  
\url{https://github.com/DavidGaona/belief_evolution_simulator.git}
\end{enumerate}


All in all, this paper highlights the impact of silence dynamics and memory on opinion formation and highlights the limitations of consensus in more nuanced  models. 

The paper is organized as follows: The new silence opinion models are introduced in Section~\ref{sec:model}. The study of consensus for these models is presented in Section~\ref{sec:convergence}. Our large-scale simulations and  case studies emerging from our models and highlighting the spiral of silence effect are presented in Section~\ref{sec:exp}. The concluding remarks are given in Section~\ref{sec:conclusion}. For the sake of space, all proofs have been moved to the Appendix.

    \section{Opinion Models}
    \label{sec:model}
    In the DeGroot framework \cite{degroot}, each agent updates their opinion by taking the weighted average of the opinions of those who influence them. The models of this framework, however, do not account for the social phenomenon known as the Spiral of Silence \cite{Noelle-Neumann1974}, where some agents may choose to become or remain silent if their opinion does not align with the majority. As a result, their \emph{current} opinion may not influence their contacts.

In this section, we generalize the DeGroot framework to take into account the Spiral of Silence. If an agent $j$ decides to be silent, there are at least two natural options when updating the opinions of the agents having $j$ as a contact: (1) agent $j$ is simply ignored in the update since their current opinion is unknown (or not public), or (2) the most recent opinion when $j$ was not silent is taken into account in the update. The former corresponds to a scenario where, for privacy purposes, opinions (messages) are removed once they have been accessed. The latter represents a typical scenario in social networks where previous opinions are kept and thus continue to influence others despite the agent's current silence.   

The above options lead us to the two generalizations of DeGroot models studied in this paper: the \emph{silence opinion memoryless}  ($\SOMminus$) models, where previous opinions are forgotten, and the \emph{silence opinion memory-based (or history-dependent)} ($\SOMplus$) models, where previous opinions are remembered. In both, agents become silent by a majority rule  for each case. Below we introduce the elements of the models.
  
\subsection{The Influence Graph}
\label{subsec:influenceGraph}

In social learning models, a \emph{community/society} is typically represented as a directed weighted graph with edges between individuals (agents) representing the direction and strength of the influence that one carries over the other. This graph is referred to as the \emph{Influence Graph}.

\begin{definition}[Influence Graph]\label{definition:agent_network}
    An ($n$-agent) \emph{influence graph} is a weighted directed graph $G = (A, E, I)$, where $A = \{1, \ldots, n\}$, $E \subseteq A \times A$, and $I:A \times A \rightarrow [0,1]$  a weight function  such that
    $I(i,j)=0 {\rm \ iff \ } (i,j) \not \in E$ and for each $i \in A$, $\sum_{j \in \N{i} \cup \{i\}} I(j,i) = 1$
where 
$ 
    \N{i} = \{ j \in A \setminus \{i\} : (j, i) \in E\}.
$
\end{definition}

The vertices in $A$ represent $n$ agents of a given community or network. The set of edges $E\subseteq A\times A$ represents the (direct) influence relation between these agents; i.e.,  $(i,j)\in E$ means that agent $i$ \emph{(directly) influences} agent $j$. The value $I(i,j)$, for simplicity written $I_{ij}$, denotes the strength of the influence: $0$ means no influence, and a higher value means stronger influence.  The normalization condition ensures that the total influence on each agent sums to 1. The set $\N{i}$ represents the \emph{neighbors} of agent $i$.

 We recall some notions from graph theory~\cite{diestel2017}. A  sequence in $E$ of the form $(i,i_1)(i_1,i_2)\ldots(i_{m-1},j)$ is a path (of length $m$) from $i$ to $j$.  The graph $G$ is said to be \emph{strongly connected} if for every pair $(i,j)$ of distinct nodes in $A$, there is a path from $i$ to $j$. A graph $G$ is a \emph{clique} if for every pair $(i,j)$ of distinct nodes in $A$,  $(i,j)\in E$.  A cycle is a path  $(i,i_1)(i_1,i_2)\ldots(i_{m-1},i)$ with all $i, i_1,\ldots i_{m-1}$ being distinct. Finally, $G$ is \emph{aperiodic} if the greatest common divisor of the lengths of its cycles is one.    

\subsection{Silence Opinion Models}
\label{subsec:SOM}

To incorporate the Spiral of Silence into the DeGroot framework, we will model the evolution of agents' opinions alongside their decisions to remain silent about a given underlying \emph{statement} or \emph{proposition}. Such a proposition could include controversial statements like, for example, \emph{``AI poses a threat to humanity''} or \emph{``pineapple belongs on pizza''}. Thus, the state of the agents (system) with respect to the proposition involves both the \emph{state of opinion} and the \emph{state of silence}.

The \textit{state of opinion} of all agents is represented as a vector in $[0,1]^{n}$. If $\Bv$ is a state of opinion, then $\Bi{i}$ denotes the opinion of agent $i$ with respect to a given proposition. If $\Bi{i}=0$ ($\Bi{i}=1$) agent $i$ completely  disagrees (agrees) with the proposition. The higher the value, the stronger the agreement.

The \textit{state of silence} is represented as a vector in $\{0,1\}^{n}$. If $\Sv$ is a state of silence, $\Si{i} = 1$ ($\Si{i} = 0$) means that agent $i$ is \emph{speaking} (is \emph{not} speaking; i.e, agent $i$ is silent).

At each time $t\in\Nat$, every agent $i \in A$ updates their opinion and their silence state. We shall use $\Bv^t$ and $\Sv^t$ to denote the state of opinion and silence at time $t$. We can now define a general Silence DeGroot opinion model as follows.  

\begin{definition}[SO Models]\label{definition:spiral_silence_opinion} 
    A \emph{Silence Opinion (SO) model}  is a tuple $(G, \BvInit, \SvInit, \mu_G)$ where $G = (A, E, I)$ is an $n$-agent influence graph, $\BvInit$ the initial state of opinion, $\SvInit$ the initial state of silence, $\mu_G: [0,1]^n \times \{0,1\}^n \times \Nat \rightarrow [0,1]^n \times \{0,1\}^n$ the state-transition function, called \emph{(state) update function}. For every $t \in \mathbb{N}$, the state of the system at time $t+1$ is $(\Bv^{t+1}, \St{}{t+1}) = \mu_G(\Bt{}{t}, \St{}{t},t).$
\end{definition}

 The update functions can be used to express any deterministic and discrete transition from one state to the next, possibly taking into account the influence graph, the current and even previous states. These functions are typically expressed by means of equations between states. In what follows, we will define particular update functions that take the spiral of silence into account.

\subsection{Spiral of Silence Models}
\label{subsec:SSM}

To build intuition we recall that the opinion update from the DeGroot model states that each agent adopts the weighted average of the opinions of the agents that directly influence them. This update can be equivalently expressed in terms of opinion differences with their neighbors, as shown in the rightmost formula of the following equation:
    
\begin{equation}
   \Bt{i}{t+1} = \sum_{j \in \N{i} \cup \{i\}} 
    I_{ji}\cdot \Bt{j}{t} \ = \  \Bt{i}{t} + \sum_{j \in \N{i}} I_{ji} \cdot  (\Bt{j}{t} - \Bt{i}{t}) \label{degroot-upd:eq}
\end{equation}
for each $i\in A$, $t \in \Nat$. Thus, in the DeGroot model each agent updates their opinion with the weighted average of the opinion differences  with their neighbors. 

We now generalize the above DeGroot update (Eq.~\ref{degroot-upd:eq}) as opinion update functions that depend not only on the current  state of opinion but also on the state of silence and, possibly, on previous states.  

\subsubsection{Memoryless Update} 
\label{subsubsec:MLU}
Our first update corresponds to the first option mentioned at the beginning of Sec.~\ref{sec:model}: The opinions of silent neighbors are ignored in the update. This can be realized by modifying Eq.~\ref{degroot-upd:eq}
as shown in the following opinion update equation: 

\begin{equation}\label{ML-bup:eq}
    \Bt{i}{t+1}  = \Bt{i}{t} + \sum_{j \in \N{i}} I_{ji} \cdot \St{j}{t} \cdot (\Bt{j}{t} - \Bt{i}{t})
\end{equation}


We now define the corresponding silence update function following the Spiral of Silence Theory. First, we need some notation. Let $x,y,\tau \in [0,1]$. The $\tau$-proximity relation $x \sim_\tau y$ holds true iff $|x-y|\leq \tau$, i.e., if $x$ and $y$ are within a tolerance radius $\Toli.$ Also, let $\N{i}^t = \{j \in \N{i} : \Si{j}^t = 1\}$ be the sets of non-silent neighbors of $i$ at time $t.$ The silence update function is given as follows:

\begin{equation}\label{ML-sup:eq}
    \St{i}{t+1}  =
    \begin{cases}
        1 & \text{if } \Maj_i  \cdot |\N{i}^t| \leq |\{ j \in \N{i}^t\ | \ \Bt{i}{t} \sim_{\Toli{i}} \Bt{j}{t} \}| \\
    0 & \text{otherwise }
    \end{cases}
\end{equation}  
For each agent~$i$, the constants $\Toli{i}, \Maj_i \in [0,1]$ represent the agent’s \emph{tolerance radius} and \emph{majority threshold} i.e., the minimum proportion of neighbors required for a group to be regarded as a majority by agent~$i$.\footnote{Although the constraint $\Maj_i \geq 0.5$ may seem natural, we neither require nor assume it in our technical results or simulations.}

Intuitively, an agent $i$ considers the opinion of $j$ to be close enough to theirs if it is within their tolerance radius $\tau_i$. Agent $i$ decides to speak iff at least a fraction $\Maj_i$, their majority threshold,  of their \emph{non-silent} contacts have opinions close enough to theirs.

We can now define the memoryless models for the spiral of silence.
\begin{definition}[$\SOMminus$] Let  $M=(G, \BvInit, \SvInit, \mu_G)$ be an SO model with $G=(A,E,I)$. Then $M$ is said to be a \emph{SO Memoryless} ($\SOMminus$) model  if for each $i\in A$ and $t\in \Nat$, 
$\mu_G(\Bt{}{t}, \St{}{t},t)=(\Bv^{t+1}, \St{}{t+1})$ where $\Bt{i}{t+1}$ and $\St{i}{t+1}$  are determined by Eq.~\ref{ML-bup:eq} and   Eq.~\ref{ML-sup:eq}, respectively. 
\label{definition:SOM-}
\end{definition}
Clearly, we can recover the DeGroot update (Eq.~\ref{degroot-upd:eq}) by setting each
tolerance radius constant $\Toli{i}$ in Eq.~\ref{ML-sup:eq} to 1 (or by setting each majority threshold $\Maj_i$ to 0) and the initial state of silence $\SvInit$ to the unit vector $\mathbf{1}_n=(1,1,\dots,1)$ of size $n$.

\begin{remark} The dynamic nature of the influence graph in $\SOMminus$ models sets them apart from the static influence in the DeGroot model. Silencing 
 an agent $j$ at a given time amounts to removing all edges $(j,i) \in E$ from the graph at that moment. This allows for more complex opinion formation behaviors.

Furthermore, we could have normalized the sum in Eq.~\ref{ML-bup:eq}
 by dividing it by $\sum_{j \in \N{i}^t} I_{ji}$ when this divisor is not equal to zero. \footnote{Notice we are using $\N{i}^t$ as index set in the summation rather than $\N{i}$.} While this would not impact our technical results in Section~\ref{sec:convergence}, it would amplify the influence of the non-silent neighbors of $i$ at time $t$, which may seem unnatural. Instead, notice that from Eq.~\ref{ML-bup:eq} we get:

\begin{equation}\label{MB-bup:eq_2}
 \Bt{i}{t+1}  = \Bt{i}{t} + \sum_{j \in \N{i}} I_{ji} \cdot \St{j}{t} \cdot (\Bt{j}{t} - \Bt{i}{t}) = (1- \sum_{j \in \N{i}} I_{ji}\cdot\St{j}{t})\cdot \Bt{i}{t}+ \sum_{j \in \N{i}} I_{ji} \cdot \St{j}{t} \cdot \Bt{j}{t}. 
\end{equation}
Thus, the influence $I_{ji}$ of a silent agent $j$ at time $t$ may be seen as increasing the weight of agent $i$'s opinion at that time. This can be interpreted as agent $i$ increasing confidence in their own opinion in the absence of external influence from agent $j$.

\end{remark}

\subsubsection{Memory-based Update.} 
\label{subsubsec:MU}
We now introduce the models corresponding to second option in the beginning of  Sec.~\ref{sec:model}: If   $j$ is silent at time $t$, the opinion update takes into account the opinion they had the last time unit   $u$ (where $u \leq t$) when they were not silent. For this to be well-defined, we assume that initially all agents are not silent; i.e., $\SvInit=\mathbf{1}_n$. 

Let $\Prime{t}_j = \max \{ u \leq t \mid \St{j}{u} = 1 \}$. The  \emph{public state of opinion} at time $t$ is a state of opinion $\PBv{t}$ such that $\PBv{t}_j =\Bt{j}{\Prime{t}_j}$ for each $j\in A.$  The following opinion  update equation captures the above intuition: 
\begin{equation}\label{MB-bup:eq}
   \Bt{i}{t+1}  = \Bt{i}{t} + \sum_{j \in \N{i}} I_{ji} \cdot (\PBv{t}_j - \Bt{i}{t})
\end{equation}

The corresponding silence update tells us that an agent $i$ becomes or remains silent at time $t+1$ precisely when the public opinion of the majority of \emph{all} their neighbors are not close enough to their own. More precisely: 
\begin{equation}\label{MB-sup:eq}
    \St{i}{t+1}  =
    \begin{cases}
        1 & \text{if } \Maj_i  \cdot |\N{i}| \leq |\{ \ j \in \N{i}\ \ | \ \  \Bt{i}{t} \sim_{\Toli{i}} \PBv{t}_j \  \}| \\
    0 & \text{otherwise }
    \end{cases}
\end{equation}  
where $\tau_i,\Maj_i  \in[0,1]$ are the \emph{tolerance radius} and the \emph{majority threshold} constants, respectively, for agent $i$.  The memory-based models are defined thus:
\begin{definition}[$\SOMplus$] Let  $M=(G, \BvInit, \SvInit, \mu_G)$ be an SO model where $G=(A,E,I)$ is an $n$-agent influence graph. Then $M$ is said to be an \emph{SO memory-based} ($\SOMplus$) model  if $\SvInit=\mathbf{1}_n$ and for each $i\in A$ and $t\in \Nat$, 
$\mu_G(\Bt{}{t}, \St{}{t},t)=(\Bv^{t+1}, \St{}{t+1})$ where $\Bt{i}{t+1}$ and $\St{i}{t+1}$  are determined by Eq.~\ref{MB-bup:eq} and   Eq.~\ref{MB-sup:eq}, resp. 
\label{definition:SOM+}
\end{definition}

The DeGroot update (Eq.~\ref{degroot-upd:eq}) is a particular case of the $\SOMplus$ opinion update (Eq.~\ref{MB-bup:eq}): We only need 
to set each
tolerance radius constant $\tau_i$ in Eq.~\ref{MB-sup:eq} to 1 since  $\SvInit$ is already required to be the unit vector of ones $\mathbf{1}_n$ in $\SOMplus$ models.

\begin{remark} The main difference between $\SOMplus$ and the DeGroot and $\SOMminus$ models is that the latter two are \emph{Markovian} processes. The next state depends on the current state but not past states. In fact, much of the tractability of  DeGroot models derives from its connection to Markov chains. Nevertheless, the next state in $\SOMplus$ models does not depend on the entire state history but just on the most recent public opinions. In the next sections, we will see the impact of this  limited amount of memory on opinion evolution.\end{remark}

    \section{Results on Consensus}
    \label{sec:convergence}
    \emph{Consensus} is a central problem in social learning models. Often, an inability to reach a consensus is a sign of polarization. In the DeGroot framework, consensus represents convergence to the same opinion value over time.  

\begin{definition}[Consensus] Let $(G, \BvInit, \SvInit,    \mu_G)$ be an \emph{SO} model with $G=(A,E,I)$. We say that the agents in $A$ converge to consensus if there exists a value $v \in [0,1]$ such that for all $ i \in A$, $\lim_{t \to \infty}\Bt{i}{t} = v.$ 
\end{definition}

Conversely, we refer to the lack of (convergence to) consensus as \emph{dissensus}, which occurs when agents fail to converge to a single opinion value.

In this section, we explore consensus for both types of  models on two different graph topologies: \emph{clique} and \emph{strongly connected} graphs. We show that consensus can only be guaranteed in arbitrary $\SOMminus$ models on clique graphs. In all other cases, consensus cannot be guaranteed due to the existence of \emph{perpetual silence} in the case of $\SOMminus$ models and public opinions in the case for $\SOMplus$ models. 

Due to space limitations, in this section we provide only sketches of the proofs for the  lemmas and theorems; the complete proofs can be found in the Appendix.   

\subsection{\texorpdfstring{$\SOMminus$ Properties}{SOMminus Properties}}

\label{subsec:SOM-Prop}
 A key property of  $\SOMminus$ models is that if all agents become silent at time $t$, they will all speak up at the very next round $t+1$.  The following lemma formalizes this property:

\begin{restatable}{lemma}{lemmaA}
    \label{lemma:limited-silent-states}
    Let $(G, \BvInit, \SvInit, \mu_G)$ be an $\SOMminus$ model with $G=(A,E,I)$. $\text{For any} \ t \in \mathbb{N}, \ \text{if } \ \St{i}{t} = 0 \ \text{for all}\ i \in A \ \text{then } \ \text{for all} \ i \in A, \  \St{i}{t+1} = 1 $.
\end{restatable}


From Lem.~\ref{lemma:limited-silent-states}, in  $\SOMminus$ models we cannot have all agents silent forever. 

\begin{corollary}\label{no-silent-forever:cor}
    Let $(G, \BvInit, \SvInit,$ $ \mu_G)$ be an $\SOMminus$ model with $G=(A,E,I)$. For every 
    $t \in \mathbb{N}$,$ \ \text{there exist}   \ i \in A$  such that $\St{i}{t} = 1$ or $\St{i}{t+1} = 1$.
    \label{corollary:no-memoryless-model-silent-forever}
\end{corollary}

We will now prove that the sequences of maximum and minimum opinion values, $\{\max (\Bt{}{t})\}_{t\in \mathbb{N}} $ and $\{\min (\Bt{}{t})\}_{t\in \mathbb{N}} $, are  (bounded) monotonically non-increasing and non-decreasing, respectively, so they must converge to some opinion value, say $U$ and $L$ with $L\leq U$. In Section \ref{subsec:SOM-Consensus}, we will prove that, under certain conditions, $U = L$, which implies that the model converges to consensus.

First, we show that the opinion values in a state are bounded by the extreme opinions in the previous state.

\begin{restatable}[Opinion Bounds]{lemma}{lemmaB}
    \label{lemma:opinion-bound-previous-extremes} 
    Let $(G, \BvInit, \SvInit, \mu_G)$ be an  $\SOMminus$ model with $G=(A,E,I)$. 
    $\text{For any} \ t \in \mathbb{N}$,
    $ min(\Bt{}{t}) \leq \Bt{i}{t+1} \leq max(\Bt{}{t})$ for all $i \in  A$.
\end{restatable}












Notice that monotonicity does not necessarily hold for the opinion values of agents. Nevertheless, it follows from 
Lem.~\ref{lemma:opinion-bound-previous-extremes} that $max(\Bt{}{t})$ is monotonically non-increasing and $min(\Bt{}{t})$ is monotonically non-decreasing with respect to $t$.

\begin{corollary}[Monotonicity of Extremes]
    Let $(G, \BvInit, \SvInit, \mu_G)$ be an $\SOMminus$ model with $G=(A,E,I)$. For all $t \in \mathbb{N},$ 
    $\ max(\Bt{}{t+1})$ $\leq max(\Bt{}{t})$ and $ \ min(\Bt{}{t+1}) \geq min(\Bt{}{t})$  .
    \label{corollary:monotonicity-extremes}
\end{corollary}

The monotonicity (and boundedness) of extremes and the Monotonic Convergence Theorem \cite{RealAnalysis}, lead us to the existence of limits for the opinion values of extreme agents.

\begin{theorem}[Limits of Extremes] Let $(G, \BvInit, \SvInit, \mu_G)$ be an $\SOMminus$ model with $G=(A,E,I)$. There must exist $\ U,\,L \in [0,1]$ such that $\ \lim_{t \to \infty} \{max(\Bt{}{t})\} = U$ and $\ \lim_{t \to \infty} \{min(\Bt{}{t})\}$ $ = L$.
    \label{theo:extremeLimits}
\end{theorem}

Henceforth, $U$ and $L$ refer to the limit values from Th. \ref{theo:extremeLimits}. Notice that if the limits for the opinion values of extreme agents are the same (i.e., $U=L$) and by the squeeze theorem \cite{RealAnalysis}, we can conclude that all the agents converge to consensus.

\subsection{\texorpdfstring{Consensus in $\SOMminus$ \ Cliques}{Consensus in SOMminus Cliques}}
\label{subsec:SOM-Consensus}

In this section, we show that consensus is guaranteed for $\SOMminus$ models whose influence graphs are cliques with at least three agents.

We consider the \emph{minimum influence} of the underlying (clique) graph $G = (A, E, I)$, defined as the constant $I_{min}=$ $\min_{(i,j) \in E} $ $I(i,j)$ \footnote{Notice that $I_{min} > 0$ since we are assuming that the underlying graph is a clique.} and  the difference between the maximum and minimum opinion values at time $t$, defined as 
$R_{t}=max(\Bt{}{t}) - \min(\Bt{}{t})$, notice that $R_{t} \in [0,1]$ at any time $t$. 

The proof strategy is based on the following key observations:
\begin{enumerate*}[label=(\emph{\alph*})]
    \item From any time $t$ onwards, there will always be non-silent agents (Cor.~\ref{no-silent-forever:cor}).
    \item The maximum and minimum opinion values, which are monotonically non-increasing and non-decreasing respectively,  must converge to some value, say $U$ and  $L$ respectively, with $L \leq U$ (Th.~\ref{theo:extremeLimits}).

     \item As the graph is a clique, the non-silent agents will influence, through the update function, all other graph agents infinitely often (Cor.~\ref{no-silent-forever:cor}).  In each such update, the span between the maximum and minimum opinions is multiplied by at most the constant $1-I_{min}<1$ (Lem.~\ref{lemma:mepsilon}). Since this contraction recurs infinitely often, the span tends to zero, which forces all opinions to converge to a common value (Th.~\ref{theorem:consensus}).
     
\end{enumerate*}

To prove consensus (i.e., $U = L$), we now demonstrate how the difference between the maximum and minimum opinion values is reduced over different time units.


\begin{restatable}[Convergent Extremes Differences]{lemma}{lemmaD} 
    \label{lemma:mepsilon}
    Let $(G, \BvInit, \SvInit, \mu_G)$ be an $\SOMminus$ model with an $n$-agent influence graph $G=(A,E,I)$ where $G$ is a clique with $ n \geq 3$. For all $m \in \mathbb{N}$, there exists  $t \in \mathbb{N}$ such that $R_{t} \leq R_{0} \cdot (1- I_{min})^m $.
\end{restatable}

As the difference between the maximum and minimum opinions reduces over time, approaching zero, we can now state our consensus result for cliques with at least three agents. 

\begin{restatable}[Consensus in $\SOMminus$ Cliques]{theorem}{theoremB}
    \label{theorem:consensus}
    Let $(G, \BvInit, \SvInit, \mu_G)$ be an  $\SOMminus$ model with an $n$-agent influence graph $G=(A,E,I)$ where $G$ is a clique. If $ n \geq 3$, then the agents in $A$ converge to consensus.
\end{restatable}




\begin{remark}    
     Notice that for $2$-agent cliques ($n=2$) consensus is not guaranteed; let us consider a clique with two agents: Agent $1$ and Agent $2$, where the opinions are $\Bt{1}{0} = 1$ and $\Bt{2}{0} = 0$, the influences are  $I(1,2)=I(2,1)=1$, the tolerance radii are $\Toli{1} = \Toli{2} = 1$ and the majority thresholds are $\Maj_1 = \Maj_2 = 0.5$. In this case, the opinion evolution of agent $1$, starting with opinion $1$, and agent $2$, starting with opinion $0$, always alternates between the values $1$ and $0$.  Figure \ref{fig:ml_opinions-clique-2-nodes} illustrates the opinion evolution of this clique. 
     \label{Remark:two-agent-clique}
\end{remark}

\begin{figure}[ht]
    \centering
    \begin{subfigure}[]{0.22\textwidth}
        \centering
        \includegraphics[width=\textwidth,scale=0.8,page=1]{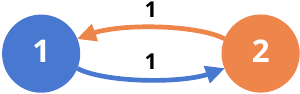}
        \caption{Influence graph}
        \label{fig:ml_opinions-clique-2-nodes-graph}
    \end{subfigure}
    \hspace{0.01\textwidth}
    \begin{subfigure}[]{0.5\textwidth}
        \centering
        \includegraphics[width=\textwidth,scale=0.8,page=1]{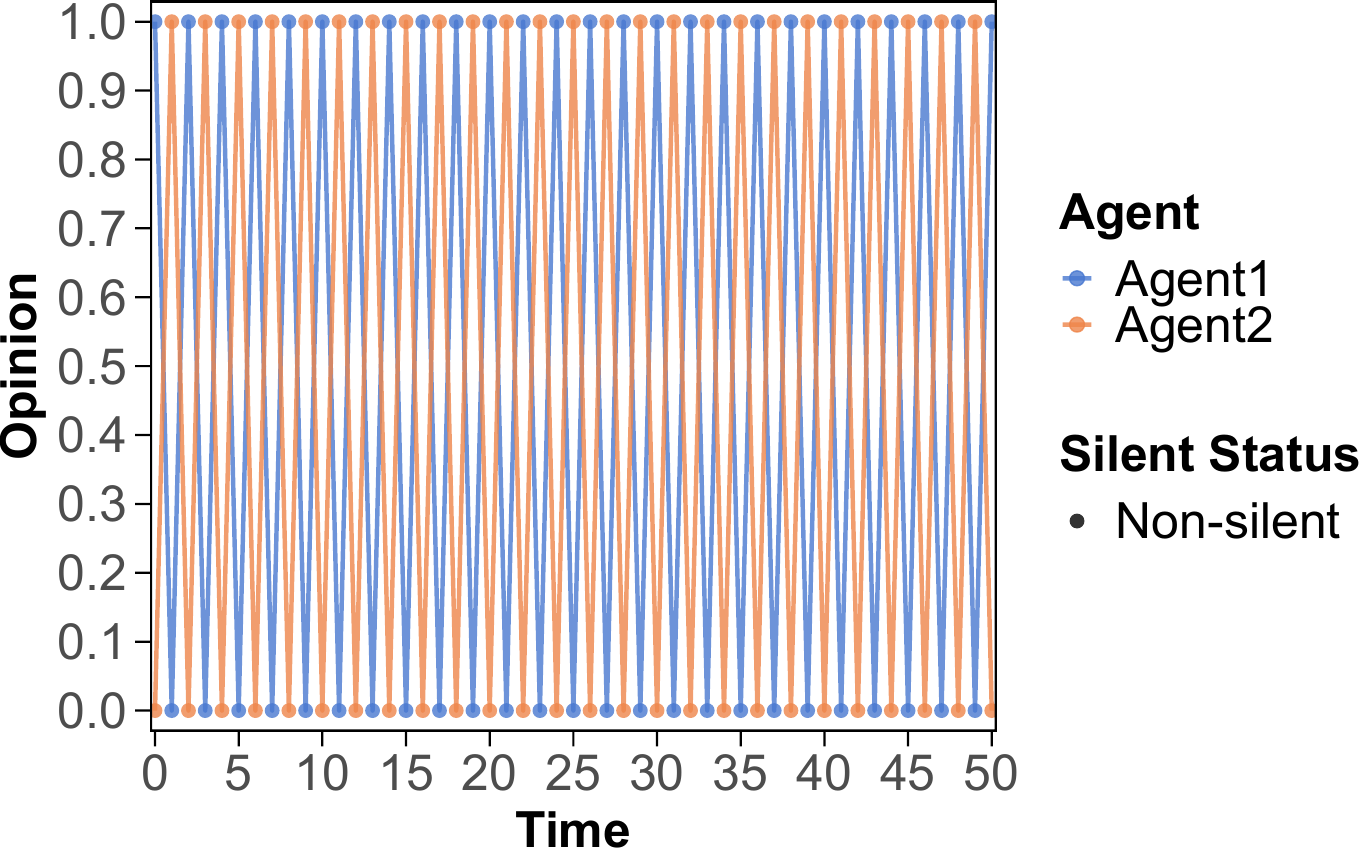}
        \caption{State evolution}
        \label{fig:ml_opinions-clique-2-nodes-state}
    \end{subfigure}
    \caption{2-agent clique opinion evolution;  $\BvInit=(1,0)$, $\tau_i=1$, $\Maj_i=0.5, i \in \{1,2\}$.}
    \label{fig:ml_opinions-clique-2-nodes}
\end{figure}

\subsection{\texorpdfstring{Dissensus in $\SOMminus$ Models}{Dissensus in SOMminus Models}}
\label{subsec:SOM-Dissensus}

In strongly connected aperiodic graphs, the $\SOMminus$ models differs from  DeGroot models by no longer guaranteeing consensus. Agents can enter a state of perpetual silence, effectively disrupting opinion propagation as if severing connections from the graph. This phenomenon is particularly critical when silent agents form bridges between connected components. Their opinions, influenced by opposing connected components, may remain below their majo threshold indefinitely. As a result, they prevent opinion exchange between components, obstructing the possibility of achieving consensus. The following example illustrates this scenario  by showing the agents state evolution and influence graph.

\begin{remark}
    For visual clarity, self-influences are omitted. As a result, the visible incoming influences for each agent may not sum to one. Nevertheless, self-influences are implicitly present to ensure conformity with Definition \ref{definition:agent_network}. Readers can infer an agent's self-influence by subtracting the sum of its visible incoming influences from one.
\end{remark}

\begin{figure}[hbtp]
    \centering
    \hspace{0.025\textwidth}
    \begin{subfigure}[t]{0.325\textwidth}
        \centering
        \includegraphics[width=\textwidth,page=1]{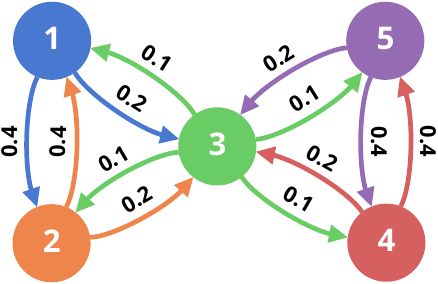}
        \caption{Strongly connected aperiodic influence graph (self influence excluded)}
        \label{fig:ml_dissensus_graph}
    \end{subfigure}
    \hspace{0.05\textwidth}
    \begin{subfigure}[t]{0.5\textwidth}
        \centering
        \includegraphics[width=\textwidth]{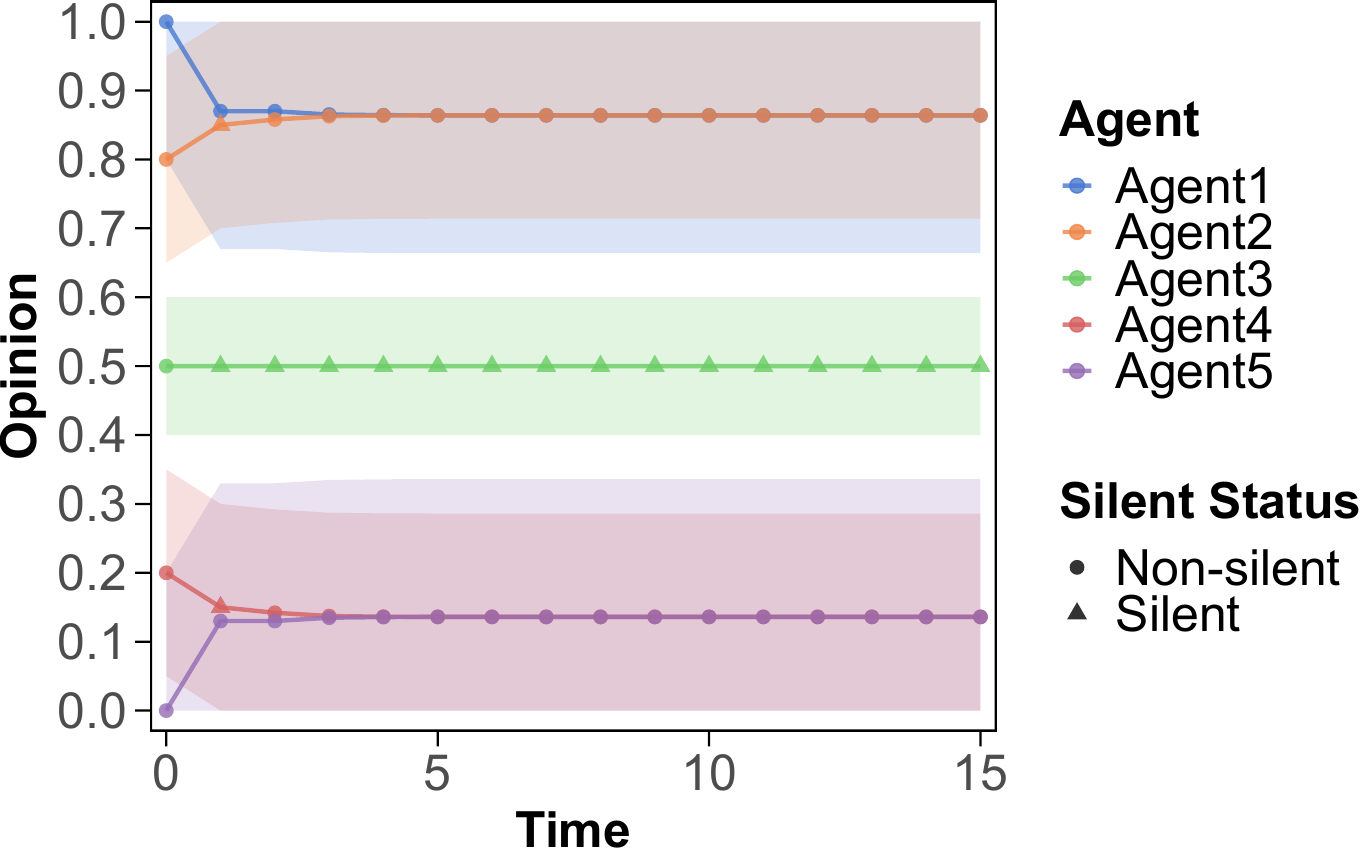}
        \caption{Each plot shows agents' state evolution over time.}
        \label{fig:ml_dissensus_chart}
     \end{subfigure}
    \hspace{0.025\textwidth}
    \caption{Examples of Dissensus in $\SOMminus$ Models. Triangles represent silent agents, circles non-silent ones. Colored areas indicate opinion values within each agent's tolerance radius. Initial state vector: $\BvInit=(1.0, 0.8, 0.5, 0.2, 0.0)$; tolerance radii $\Tol=(0.2, 0.15, 0.1, 0.15, 0.2)$; majority thresholds equal to $0.5$ for each agent.}
    \label{fig:ml_dissensus}
\end{figure}

In the example shown in Fig.~\ref{fig:ml_dissensus}, Agent 3 serves as the sole bridge connecting two distinct graph components: \{Agents 1, 2\} and \{Agents 4, 5\}. Due to the tolerance radius and majority threshold parameters, Agent 3 becomes effectively isolated due to its central position. Specifically, Agent 3's belief ($0.5$) lies outside the tolerance radius of agents in both components, causing it to have no neighbors with sufficiently similar opinions. Consequently, Agent 3 becomes perpetually silent after the first round.

Meanwhile, the opposing influences from both components on Agent 3's belief update equation cancel each other out the ``pull'' from the left component (Agents 1, 2) is balanced by an equal and opposite ``pull'' from the right component (Agents 4, 5). This results in Agent 3's belief remaining fixed while it stays silent. In contrast, agents within each component maintain beliefs within each other's tolerance radii (after the first round), ensuring they continue to speak and influence one another. As a result, each component converges to its own local consensus, while the perpetually silent Agent 3 can no longer transmit influence between them. The graph thus becomes effectively disconnected, preventing global consensus despite the original strong connectivity.

\subsection{\texorpdfstring{Dissensus in $\SOMplus$ Models}{Example of Dissensus in SOMplus Models}}
\label{subsec:SOM+Dissensus}

We now prove through an insightful counterexample that, unlike in the \emph{memoryless}   $\SOMminus$ models, consensus is not guaranteed in $\SOMplus$ models for clique graphs. This negative result shows that incorporating a minimal yet natural memory mechanism (where only the most recent public opinions are considered) can actually hinder  consensus.

While $\SOMplus$ models share similarities with $\SOMminus$ ones regarding perpetual silence, it introduces a unique phenomenon where the entire graph can enter and indefinitely remain in a silent state. 
This distinction stems from how silent agents influence the opinion update in $\SOMplus$. 

Recall from Section \ref{subsubsec:MU} that $\PBv{t}$ is the public state of opinion at time $t$ and represents the most recent public opinion of each agent. Henceforth, we will refer to $\PBv{t}_i$ as the \emph{public opinion of agent $i$}. In $\SOMplus$, speaking agents influence the opinions of their neighbors as in $\SOMminus$. Nevertheless, unlike in $\SOMminus$,  silent agents influence  their neighbors with their public opinion. 


Hence, if agent $i$'s opinion converges to a value where most recent public opinions of more than a fraction $\Maj_i$ of its neighbors fall outside its tolerance range, the agent will remain perpetually silent. This leads to scenarios where agents withdraw from discourse, while other agents continue to be influenced by these outdated public opinions.

Unlike $\SOMminus$, where opinions of silent agents are disregarded, $\SOMplus$ allows for the persistence of unchanged public opinions indefinitely. This can result in dissensus (i.e., lack of convergence to a consensus) due to the formation of public opinions that no longer reflect the current opinions of silent agents.
\subsubsection{Counter-example to Consensus.} 
Consider the $\SOMplus$ model $M=(G, \BvInit, \SvInit, \mu_G)$ with $G$ as the Clique in Fig.\ref{fig:mutiple_value_convergence_graph} and 
$\BvInit=(1.0, 0.9, 0.1, 0.0)$, tolerance radius $\Tol_i = 0.1$, majority threshold $\Maj_i=0.5$ for each agent $i$ in $G$. All agents initially have more than half of their neighbors' opinions outside their tolerance radius. Consequently, all agents become silent at $t=1$. At this point, each agent's updated opinion places all other agents' opinions outside its tolerance radius. The agents' opinions then converge to distinct values, determined by the initial opinions $\BvInit$, which perpetuates the condition for silence. This state persists indefinitely, as the convergence values maintain the silence condition for all agents. This results in dissensus as illustrated in Fig.\ref{fig:mutiple_value_convergence_chart} . 

The above counter-example demonstrates how the $\SOMplus$ models can lead to complete silence and opinion divergence in a clique, preventing the possibility of consensus.

\begin{figure}[htbp]
    \centering
    \hspace{0.05\textwidth}
    \begin{subfigure}[t]{0.3\textwidth}
        \centering
        \includegraphics[width=\textwidth]{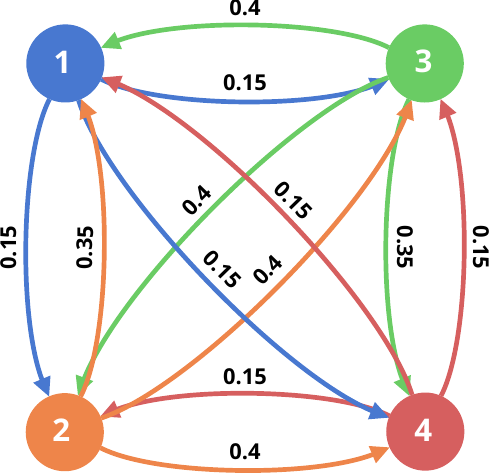}
        \caption{Clique influence graph (self influence excluded)}
        \label{fig:mutiple_value_convergence_graph}
    \end{subfigure}
    \hspace{0.075\textwidth}
    \begin{subfigure}[t]{0.5\textwidth}
        \centering
        \includegraphics[width=\textwidth]{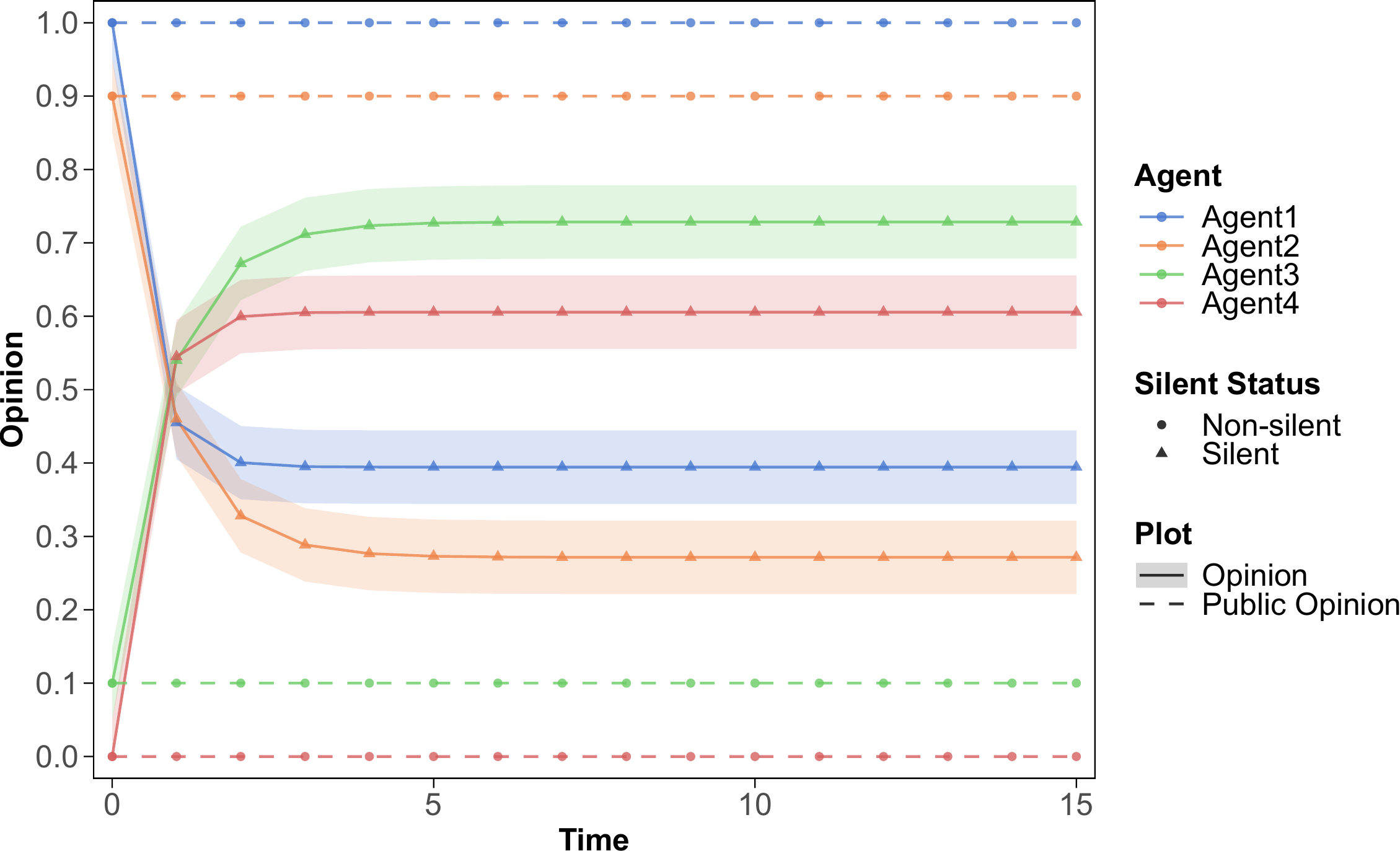}
        \caption{Each plot shows agents' state evolution over time.}
        \label{fig:mutiple_value_convergence_chart}
    \end{subfigure}
    \hspace{0.075\textwidth}
    \caption{Counter-Example to Consensus in an $\SOMplus$ Model. Triangles represent silent agents, circles non-silent ones. Colored areas indicate opinion values within each agent's tolerance radius. Initial state vector: $\BvInit=(1.0, 0.9, 0.1, 0.0)$, tolerance radius and majority threshold equal to $0.1$ and $0.5$ for each agent.}
    \label{fig:mutiple_value_convergence}
\end{figure}

\section{Experimenting with the Spiral of Silence}
\label{sec:exp}

This section demonstrates how  $\SOMminus$ and  $\SOMplus$ models can capture key dynamics of the spiral of silence through agent-based simulations. We first showcase an $\SOMminus$ model that reinforces dominant opinions, enabling vocal minorities to disproportionately influence public discourse despite numerical inferiority. Next, we illustrate an $\SOMplus$ model revealing an apparent dichotomy: agents privately converge to \emph{consensus} while maintaining \emph{divergent} public stances, mirroring social  environments where outdated expressions mask underlying agreement \cite{Hampton2014}. 

To validate scalability, we simulate networks ranging from 4 to over 2.1 million agents ($2^2$ to $2^{21}$) with varying connectivity, leveraging a modified version of Ross et al.'s algorithm \cite{Ross2019} to generate networks with power-law ("rich-get-richer") degree distributions and small-world properties. Small-world properties are characterized by short average path lengths between nodes while maintaining high local clustering. This reflects \emph{real social networks}' structure where information spreads efficiently despite limited local connections, making our simulations more realistic for studying opinion dynamics. Indeed, studies have shown social networks exhibit 6 degrees of separation \cite{Travers1969}, that is, any one person is on average 6 connections away from any other. 

\subsection{ Reinforcement of Dominant Views}

Figure \ref{fig:partial_engagement} illustrates how $\SOMminus$ models capture the essence of the Spiral of Silence theory. In this strongly connected graph, a vocal minority effectively dominates the discourse, causing the silent majority to converge toward the perceived majority opinion.

The network has two groups: Group 1 (Agents 1, 2, 3, 6, 7, and 8) with opinions in the lower half of the spectrum, and Group 2 (Agents 4 and 5) with opposing views near the higher end. Despite Group 1 being the actual majority, the network topology and differing tolerance radii $\Tol$ lead to Group 2 dominating the opinion dynamics. Group 1 agents lack intra-group connections but are all linked to Group 2 agents. Also, Group 2's significantly larger tolerance radii allow them to remain non-silent more frequently.

\begin{figure}[htbp]
    \centering
    \begin{subfigure}[b]{0.45\textwidth}
        \centering
        \includegraphics[width=\textwidth]{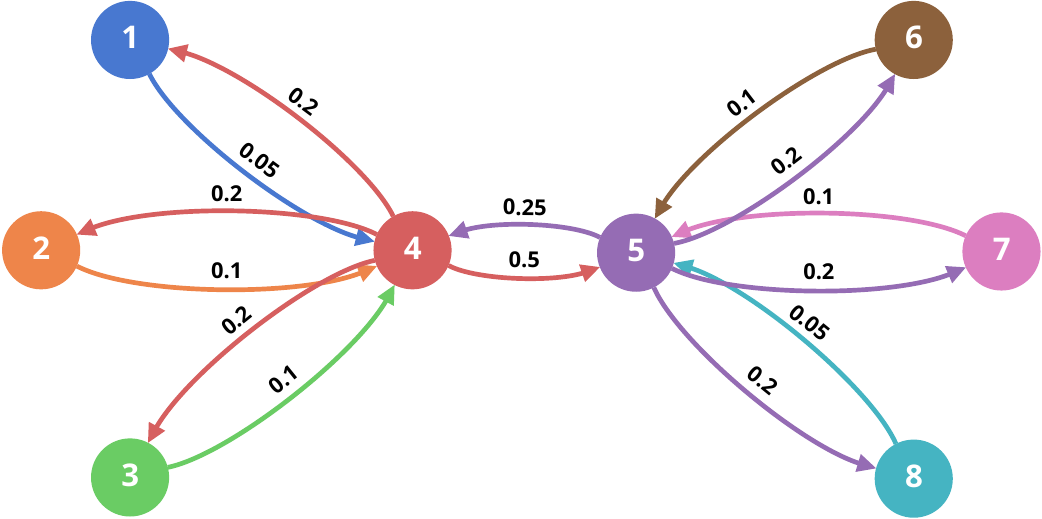}
        \caption{Aperiodic strongly connected influence graph (self influence excluded)}
        \label{fig:partial_engagement_graph}
    \end{subfigure}
    \hfill
    \begin{subfigure}[b]{0.5\textwidth}
        \centering
        \includegraphics[width=\textwidth]{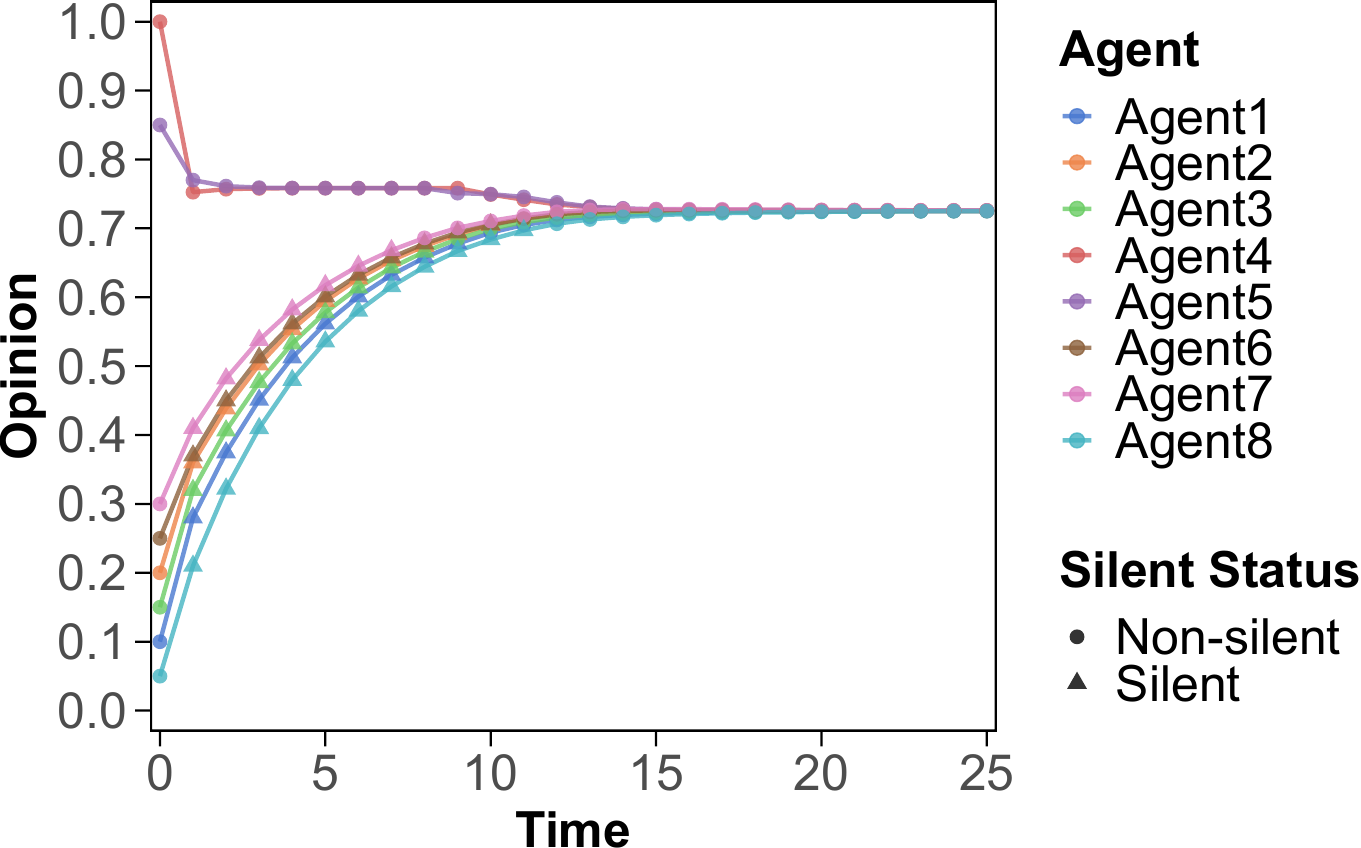}
        \caption{Each plot shows agents' state evolution over time.}
        \label{fig:partial_engagement_chart}
    \end{subfigure}
    \caption{Silent Majority vs. Vocal Minority: Opinion Dynamics in $\SOMminus$. Triangles represent silent agents, circles non-silent ones. Colored areas indicate opinion values within each agent's tolerance radius. Here $\BvInit = (0.1, 0.2, 0.15, 1.0, 0.85, 0.25, 0.3, 0.05)$ and $\Tol = (0.1, 0.05, 0.1, 0.85, 0.6, 0.05, 0.1, 0.05)$ and majority thresholds are set to $0.5$.}
    \label{fig:partial_engagement}
\end{figure}

This configuration results in the vocal minority (Group 2) disproportionately influencing the network. The silent majority (Group 1) remains quiet for most of the update process, only becoming active when opinions have already shifted closer to the perceived majority view. This example demonstrates how the $\SOMminus$ models can simulate scenarios where a minority opinion, through strategic positioning and persistent vocalization, can shape the overall opinion landscape, even when numerically outnumbered.

\subsection{\texorpdfstring{Hidden Consensus in $\SOMplus$ Models}{Hidden Consensus in SOMplus Models}}

$\SOMplus$ models reveal a noteworthy phenomenon: the possibility of reaching a consensus that remains undetected by the agents themselves. Figure \ref{fig:public_dissensus_opinion_consensus} illustrates this scenario using a clique graph with four agents.

Initially, the agents hold diverse opinions ($\BvInit=(1.0, 0.9, 0.1, 0.0)$). However, due to the graph's influence structure, all agents become silent after $t=0$. Despite this silence, their private opinions converge to a common value  over time. Crucially, this convergence occurs without any further public expression of opinions, leaving each agent unaware of the emerging consensus. This hidden consensus phenomenon mirrors real-world scenarios in social media where individuals may unknowingly share common views while perceiving disagreement due to outdated public expressions \cite{Hampton2014}. 

\begin{figure}[ht]
    \centering
    \hspace{0.0675\textwidth}
    \begin{subfigure}[t]{0.32\textwidth}
        \centering
        \includegraphics[width=\textwidth]{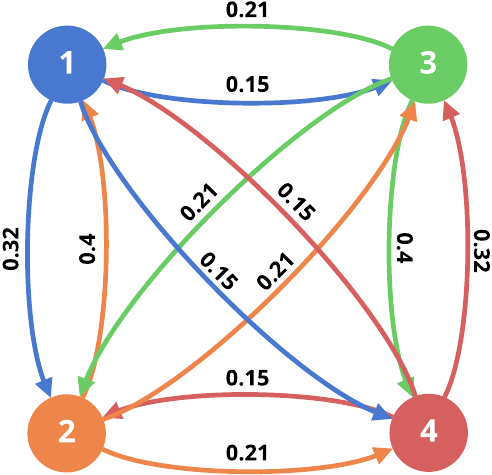}
        \caption{Influence graph self influence excluded}
        \label{fig:public_dissensus_opinion_consensus_graph}
    \end{subfigure}
    \hspace{0.05\textwidth}
    \begin{subfigure}[t]{0.5\textwidth}
        \centering
        \includegraphics[width=\textwidth]{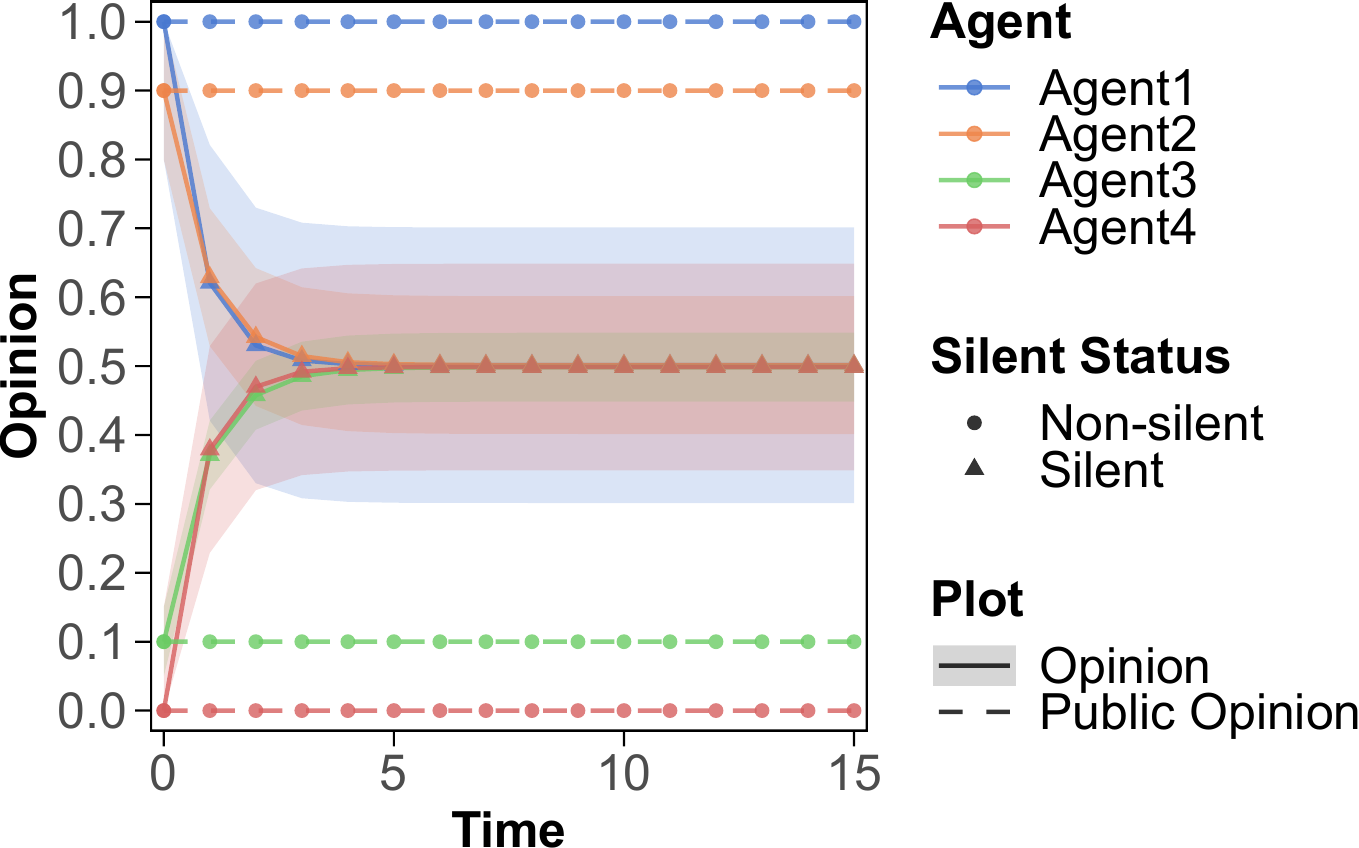}
        \caption{Each plot shows agents' state evolution over time.}
        \label{fig:public_dissensus_opinion_consensus_chart}
    \end{subfigure}
    \hspace{0.0675\textwidth}
    \caption{Silent Consensus Scenario. Triangles represent silent agents, circles non-silent ones. Colored areas indicate opinion values within each agent's tolerance radius. Initial state vector: $\BvInit=(1.0, 0.9, 0.1, 0.0)$; tolerance radii $\Tol=(0.2, 0.1, 0.05, 0.15)$; majority thresholds equal to $0.5$ for each agent.}
    \label{fig:public_dissensus_opinion_consensus}
\end{figure}

\subsection{Scaling to Large-Scale Network Simulations}
\label{subsec:scaling}

While small-scale simulations illustrate core mathematical behaviors, they do not capture the complexity of real-world social networks. To bridge this gap, we developed a scalable simulation platform supporting networks of over 2.1 million agents ($2^{21}$), enabling analysis of emergent phenomena in realistic settings.

Our implementation combines a modified preferential attachment algorithm \cite{Barabási&Albert2002} to generate networks with small-world properties \cite{Watts1998} and power-law (rich-get-richer) degree distributions. To address computational intensity, we employ parallel processing via Scala and Akka Actors \cite{Hewitt1973,Hewitt2010}, with results persisted in PostgreSQL for efficient querying. For each experiment, we generated networks with \emph{density} (minimum neighbor count) ranging from 1 to 15, executing 1,024 simulations per parameter combination. Consensus rates were then calculated across all runs for statistical significance. All agents have their tolerance radius and majority threshold set to $0.1$ and $0.5$, respectively, and their initial beliefs as random uniformly distributed values from 0 to 1. The code is available  \href{https://github.com/DavidGaona/belief_evolution_simulator.git}{here}.

\subsubsection{Contrasting Silence Opinion Models.}

\label{subsubsec:SOM-comparison-results}
The $\SOMplus$ and $\SOMminus$ models exhibit dramatically opposing behaviors in consensus formation as seen in Fig.\ref{fig:SOM-heatmap}. For $\SOMplus$ models, consensus becomes increasingly rare beyond 512 agents ($2^9$), occurring in only 0.097\% of simulations (1/1024 runs) at this threshold, exclusively for density~4. Beyond this scale, no simulations reached consensus. This aligns with social media dynamics where persistent disagreement emerges, as agents' historical opinions create perception gaps between private consensus and public expression \cite{Hampton2014}. Even a minimal notion of memory (retaining only  most recent public opinions) proves sufficient to sustain dissensus. 

\begin{figure}[htbp]
    \centering
    \includegraphics{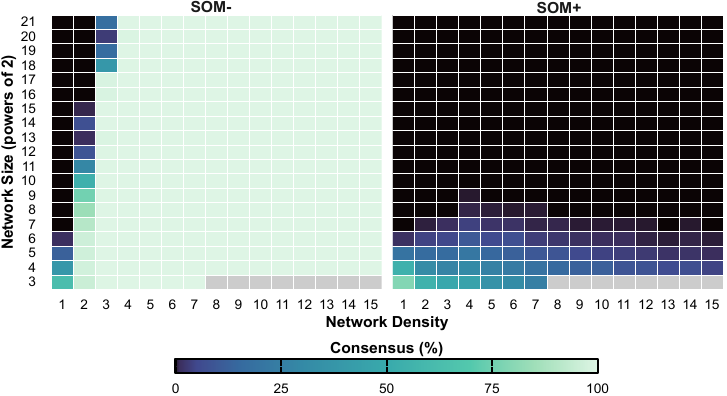}
    \caption{Heatmap indicating consensus \% across different network sizes (y-axis) and densities (x-axis) for both $\SOMplus$ and $\SOMminus$ models. Lighter hues indicate a higher rate of consensus, while darker hues indicate a lower rate. All tolerance radii and majority thresholds are equal to $0.1$ and $0.5$ resp.}
    \label{fig:SOM-heatmap}
\end{figure}

Conversely, $\SOMminus$ models exhibit the opposite scalability behavior. Once a certain density threshold is reached, consensus consistently emerges in all simulations. Larger networks require slightly higher connectivity: for 2,048 agents ($2^{11}$), a density of 3 suffices, while for 2.1 million agents ($2^{21}$), a density of 4 is needed. This highlights how the spiral of silence depends on network structure. In sparse networks, some agents may remain silent indefinitely due to critical bridge nodes (see Section~\ref{subsec:SOM-Dissensus}). In contrast, denser networks make it increasingly unlikely for bridges to be disconnected, since disrupting consensus would require silencing more than $d$ agents (for density $d$)—a rare occurrence in networks of this type.

    \section{Conclusions and Related Work}
    \label{sec:conclusion}
    We  extended the classical DeGroot framework to incorporate key social dynamics described by the Spiral of Silence. Our contributions highlight how the addition of memory, even in a limited form, as in $\SOMplus$ models, introduces significant complexity to consensus-building processes. We demonstrated that consensus, while achievable for cliques in $\SOMminus$ models, is no longer guaranteed, in contrast with  DeGroot models,  in arbitrary strongly-connected aperiodic graphs. This points to the impact that silence and memory can have on opinion formation in social networks. It also offers insights into the challenges of converging to consensus in real-world scenarios. We also discuss simulations reflecting predictions of the Spiral of Silence, such as the reinforcement of dominant views and hidden consensus. Finally, we discussed the simulation in our models of large-scale graphs reflecting the small-world topology of real social networks. 

\emph{Future Work}. Experimentation and validation on real social networks typically involve continuously collecting user opinions and inferring influence via methods such as stance detection and influencer identification from natural language processing, machine learning and graph mining.
We defer this challenging task to future work.


\emph{Related Work}. To our knowledge, no prior work has extended DeGroot-based models to incorporate the Spiral of Silence. Some recent studies have examined the Spiral of Silence in agent-based networks. E.g.,  \cite{Ross2019} models manipulative actors in social networks using agents with fixed binary opinions (agree/disagree) who choose whether to express them based on the prevailing opinion climate. In \cite{Cabrera2021}, the authors explore how the number of communities and their connectivity influence perceived opinion climates. However, these studies do not address opinion updates, convergence to consensus, or memory, which are central aspects of this paper.

In \cite{chistikov2020convergence,zehmakan2021majority} the stabilization property was studied in majority-based opinion models, where individuals adopt the most prevalent opinion in their social circles (their neighbors). Although these works focus on majority-based opinion dynamics, similar to models inspired by Spiral of Silence, their models are exclusively memoryless with only two-value discrete opinions.

Recent studies \cite{biondi2023dynamics,shirzadi2024stubborn,xu2022effectsstubbornnessopiniondynamics} explored conflict, disagreement, and polarization using the Friedkin and Johnsen (FJ) model. This model extended DeGroot’s by adding a fixed internal opinion and modeling stubbornness, where agents’ expressed opinions reflected their internal beliefs. While silent agents in $\SOMplus$ and stubborn agents in the FJ model shared some similarities, their behavior differed. In $\SOMplus$, silent agents could later become active, allowing their public opinions to change. In contrast, stubborn agents in the FJ model did not shift in this way. Thus, although neither model guaranteed consensus, their opinion dynamics diverged.

The exclusion of silent agents in $\SOMminus$ models amounts to having edges (influences) disappearing and reappearing during opinion evolution which reflects the dynamic influence nature of this model. There are several works studying dynamic influence in opinion formation. The work \cite{demarzo2003persuasion} introduces a version of the DeGroot model in which self-influence changes over time while the influence on others remains the same. The works \cite{chatterjee1977towards,Generalize2} explore convergence and stability, respectively, in models where influences change over time.  The work \cite{Aranda2024}   demonstrates how asynchronous communication, when combined with dynamic influence, can prevent consensus.  None of this work deals with the dynamics derived from the Spiral of Silence.

\newpage

    \bibliographystyle{splncs04}
    \bibliography{references.bib}
\newpage
    \appendix
    \label{Appendix}
    \section{Appendix}

\label{app:proofs}

\lemmaA*
\begin{proof}
    Let $\St{}{t}$ such that $\St{i}{t} = 0 \ \text{for all} \ i \in A$.  Then, by applying  Eq.~\ref{ML-sup:eq}, we can conclude that $\St{i}{t+1} = 1 \ \text{for all} \ i \in A$, given that $|\N{i}^t| = 0 \ \text{for all} \ i \in A$, 
\end{proof}

\lemmaB*
\begin{proof}
    From Definition \ref{definition:SOM-} and Eq.~\ref{MB-bup:eq_2} we know that:


$$ \Bt{i}{t+1} = (1- \sum_{j \in \N{i}} I_{ji}\cdot\St{j}{t})\cdot \Bt{i}{t}+ \sum_{j \in \N{i}} I_{ji} \cdot \St{j}{t} \cdot \Bt{j}{t} 
$$


Now, as $\Bt{k}{t} \leq max(\Bt{}{t})$  for all $k$ $\in$ $A$, we can conclude that:


    $$ 
        \Bt{i}{t+1} 
        \leq (1- \sum_{j \in \N{i}} I_{ji}\cdot\St{j}{t})\cdot max(\Bt{}{t})+ \sum_{j \in \N{i}} I_{ji} \cdot \St{j}{t} \cdot max(\Bt{}{t})
        = max(\Bt{}{t})
    $$


    Thus, $\Bt{i}{t+1} \leq max(\Bt{}{t})$ as wanted. The proof that $min(\Bt{}{t}) \leq \Bt{i}{t+1}$ is similar. 
\end{proof}

\lemmaD*
\begin{proof} 
 First, we prove that at each time unit $t$ where there is at least one non-silent agent, the distance between the maximum and minimum opinions will be reduced by at least a factor of  $(1-  I_{min}) $ at time $t +1$, i.e. $R_{t+1} \leq (1- I_{min} ) \cdot  R_{t}$  where $1- I_{min}$ is a positive constant below one.
 

We consider the following two cases for each unit time $t$ where there is at least one non-silent agent: 


\begin{enumerate}
\item At least one agent is silent. In this case, we consider a non-silent agent $k$, the updated opinion of every agent $i$ influenced by agent $k$ will have an upper bound, thus: 

$$\Bt{i}{t+1} = (1- \sum_{j \in \N{i}} I_{ji}\cdot\St{j}{t})\cdot \Bt{i}{t}+ \sum_{j \in \N{i}} I_{ji} \cdot \St{j}{t} \cdot \Bt{j}{t}  \leq (1-I_{min})  \cdot \max(\Bt{}{t})+ I_{min} \cdot \Bt{k}{t} $$

And the lower bound would be the following:

$$\Bt{i}{t+1} = (1- \sum_{j \in \N{i}} I_{ji}\cdot\St{j}{t})\cdot \Bt{i}{t}+ \sum_{j \in \N{i}} I_{ji} \cdot \St{j}{t} \cdot \Bt{j}{t} \geq (1- I_{min} )\cdot \min(\Bt{}{t}) + I_{min} \cdot \Bt{k}{t} $$

Therefore: 
\begin{equation}
(1-I_{min})  \cdot \max(\Bt{}{t})+ I_{min} \cdot \Bt{k}{t} \geq  \Bt{i}{t+1} \geq (1- I_{min} )\cdot \min(\Bt{}{t}) + I_{min} \cdot \Bt{k}{t} 
\label{degroot-upd:eqlemma4_1_1_1}
\end{equation}

As the graph is a clique and there is at least one silent agent,  every agent in the graph is influenced by agent $k$, including itself; notice that as there is at least one silent agent, call it $l$, then $(1- \sum_{j \in \N{i}} I_{ji})$  from Eq.~\ref{MB-bup:eq_2}, corresponding to the self-influence on agent $i$,  is at least $I_{li} \geq  I_{min}$. Therefore, all agents in the graph satisfy    Eq.~$\ref{degroot-upd:eqlemma4_1_1_1}$; then, there are both an  upper and lower  bound for all opinions in unit time $t+1$, including the maximum and the minimum opinions, thus:

\begin{equation}
\max(\Bt{}{t+1})  \leq (1-I_{min})  \cdot \max(\Bt{}{t})+ I_{min} \cdot \Bt{k}{t} 
\label{degroot-upd:eqlemma4_2}
\end{equation}

\begin{equation}
 \min(\Bt{}{t+1}) \geq (1- I_{min} )\cdot \min(\Bt{}{t}) + I_{min} \cdot \Bt{k}{t} 
\label{degroot-upd:eqlemma4_3}
\end{equation}



Then, as  $R_{t+1} = \max(\Bt{}{t+1}) - \min(\Bt{}{t+1}) $, $R_{t+1} \leq (1-I_{min})  \cdot \max(\Bt{}{t})+ I_{min} \cdot \Bt{k}{t}  - (1- I_{min} )\cdot \min(\Bt{}{t}) - I_{min} \cdot \Bt{k}{t}  = (1- I_{min} ) \cdot   (\max(\Bt{}{t}) - \min(\Bt{}{t})) = (1- I_{min} ) \cdot  R_{t}$.  Then  $R_{t+1} \leq (1- I_{min} ) \cdot  R_{t}$.

\item All agents speak at time $t$: as the graph is a clique with $ n \geq 3$, there must be at least one non-silent agent $k$ that is not the unique maximum or minimum opinion agent in the graph \footnote{Notice that this case is not applicable in the clique in Remark \ref{Remark:two-agent-clique}, as such a clique does not have a non-silent agent that is not the unique maximum or minimum opinion agent in the graph.}; agent $k$  influences all the other agents such that their opinions are bounded as follows:

$$\Bt{i}{t+1} = (1- \sum_{j \in \N{i}} I_{ji}\cdot\St{j}{t})\cdot \Bt{i}{t}+ \sum_{j \in \N{i}} I_{ji} \cdot \St{j}{t} \cdot \Bt{j}{t} \leq (1-I_{min})  \cdot \max(\Bt{}{t})+ I_{min} \cdot \Bt{k}{t} $$

$$\Bt{i}{t+1} = (1- \sum_{j \in \N{i}} I_{ji}\cdot\St{j}{t})\cdot \Bt{i}{t}+ \sum_{j \in \N{i}} I_{ji} \cdot \St{j}{t} \cdot \Bt{j}{t} \geq (1- I_{min} )\cdot \min(\Bt{}{t}) + I_{min} \cdot \Bt{k}{t} $$

As agent $k$ is influenced by an agent with maximum opinion and an agent with minimum opinion at time $t$, its opinion at time $t+1$ will be bounded as follows:

$$\Bt{k}{t+1} \geq (1- I_{min}) \cdot \min(\Bt{}{t}) + 
I_{min} \cdot \max(\Bt{}{t})$$

$$\Bt{k}{t+1} \leq (1- I_{min}) \cdot \max(\Bt{}{t}) + 
I_{min} \cdot \min(\Bt{}{t})$$

Therefore: 

$$\Bt{k}{t+1} \geq (1- I_{min}) \cdot \min(\Bt{}{t}) + 
I_{min} \cdot \max(\Bt{}{t}) \geq (1- I_{min}) \cdot \min(\Bt{}{t}) + 
I_{min} \cdot \Bt{k}{t}$$

and 

$$\Bt{k}{t+1} \leq (1- I_{min}) \cdot \max(\Bt{}{t}) + 
I_{min} \cdot \min(\Bt{}{t})  \leq (1- I_{min}) \cdot \max(\Bt{}{t}) + 
I_{min} \cdot \Bt{k}{t}$$

Therefore, every agent in the graph, including $k$, holds the following bounds:
 
\begin{equation}
(1-I_{min})  \cdot \max(\Bt{}{t})+ I_{min} \cdot \Bt{k}{t} \geq  \Bt{i}{t+1} \geq (1- I_{min} )\cdot \min(\Bt{}{t}) + I_{min} \cdot \Bt{k}{t} 
\label{degroot-upd:eqlemma4_1_1}
\end{equation}

Then, as  $R_{t+1} = \max(\Bt{}{t+1}) - \min(\Bt{}{t+1}) $, $R_{t+1} \leq (1-I_{min})  \cdot \max(\Bt{}{t})+ I_{min} \cdot \Bt{k}{t}  - (1- I_{min} )\cdot \min(\Bt{}{t}) - I_{min} \cdot \Bt{k}{t}  = (1- I_{min} ) \cdot   (\max(\Bt{}{t}) - \min(\Bt{}{t})) = (1- I_{min} ) \cdot  R_{t}$.  Then  $R_{t+1} \leq (1- I_{min} ) \cdot  R_{t}$.

\end{enumerate}

Thus, we have proved that $R_{t+1} \leq (1- I_{min} ) \cdot  R_{t}$ for any time unit $t$ where there is at least one non-silent agent. 

Finally, by using the above and  
Cor.~\ref{corollary:no-memoryless-model-silent-forever}, it can be proved, by induction,  that for all $m \in \mathbb{N}$, there exists  $t \in \mathbb{N}$ such that $R_{t} \leq R_{0} \cdot (1- I_{min})^m $.

\end{proof}

\theoremB*
\begin{proof}
    From Th.~\ref{theo:extremeLimits}, there exist $U,\,L \in [0,1]$ such that $U= \lim_{t \to \infty} max(\Bt{}{t})$ and $L =  \lim_{t \to \infty} min(\Bt{}{t})$. We now prove $U=L$ by contradiction. Suppose that $U \neq L$.  As $U$ can not be lower than $L$, therefore, $U - L > 0$.   From $\lim_{m \to \infty} R_{0} \cdot (1- I_{min})^m  = 0$ \footnote{Notice that $R_{0}$ and  $I_{min}$ are positive constants such that  $ 1 > (1 - I_{min}) > 0$.} and  Lem.~\ref{lemma:mepsilon} there must exists a time $t$ where $ R_{t}  <  (U- L)$,  but it is not possible as, from Th.~\ref{theo:extremeLimits}  and Cor. ~\ref{corollary:monotonicity-extremes}, $\min(\Bt{}{t})\leq  L $
     and $\max(\Bt{}{t})\geq  U $,  which is a contradiction. Therefore, $U=L$ and the model converges to consensus.
\end{proof}

\end{document}